\documentclass[a4paper]{article}
\usepackage[utf8]{inputenc}
\usepackage{thm-restate}
\usepackage[margin=2cm]{geometry}
\usepackage{authblk}

\usepackage[T1]{fontenc}
\usepackage{graphicx} 
\usepackage[ruled,noend]{algorithm2e}
\usepackage{algpseudocode}
\usepackage{amsthm}
\usepackage{amssymb,amsmath}  
\usepackage{mathtools}
\usepackage[colorlinks=true, allcolors=blue!80!black, linktocpage]{hyperref}
\usepackage{bookmark}
\usepackage[capitalize]{cleveref}
\usepackage{todonotes}
\usepackage{placeins}
\usepackage{booktabs}
\usepackage{microtype}

\usepackage{comment}

\usepackage{newtxtext}
\usepackage{newtxmath}

\usepackage{enumitem}
\setlist[1]{labelindent=\parindent}
\setlist[enumerate]{label=(\arabic*)}
\setlist[itemize]{noitemsep}

\usepackage{xcolor}

\newcommand{\setsize}[1]{\left|#1\right|}

\newcommand{\N}{\mathbb{N}}

\newcommand{\prefdeg}{\texttt{prefdeg}}
\newcommand{\suffdeg}{\texttt{suffdeg}}
\newcommand{\len}{\texttt{len}}

\definecolor{CWIBlue}{rgb}{0.38, 0.478, 0.56}
\definecolor{CWIDGreen}{rgb}{0.372, 0.494, 0.419}
\definecolor{CWILGreen}{rgb}{0.603, 0.717, 0.486}
\definecolor{CWIRed}{rgb}{0.686, 0.188, 0.262}

\newif\ifhidetodos
\hidetodostrue

\usepackage{marginnote}

\newcommand{\cbox}[2][yellow]{%
  \fcolorbox{#1}{white}{\parbox{\dimexpr\linewidth-2\fboxsep}{\strut #2\strut}}%
}

\newcommand{\customtodo}[3]{\textcolor{#2}{\(\blacktriangledown\)}\marginnote{\raggedright \textcolor{#2}{\textbf{#1:} #3}}}
\newcommand{\customtododisplay}[3]{\noindent\cbox[#2]{\textcolor{#2}{\textbf{#1:} #3}}}
\newcommand{\customtodoinline}[3]{\textcolor{#2}{\textcolor{#2}{\textbf{#1:} #3}}}

\ifhidetodos
    \renewcommand{\customtodo}[3]{}
    \renewcommand{\customtododisplay}[3]{}
    \renewcommand{\customtodoinline}[3]{}
\fi

\newtheorem{theorem}{Theorem}[section]
\newtheorem{corollary}[theorem]{Corollary}
\newtheorem{lemma}[theorem]{Lemma}
\newtheorem{definition}[theorem]{Definition}

\newcommand{\exopt}{\mathsf{exopt}}

\title{Greedy Algorithms for Shortcut Sets and Hopsets}
\author[1,2]{Ben Bals}
\author[1]{Joakim Blikstad}
\author[3]{Greg Bodwin}
\author[1]{Daniel Dadush}
\author[4]{Sebastian Forster}
\author[1,2]{Yasamin Nazari}
\affil[1]{CWI, Amsterdam, The Netherlands}
\affil[2]{Vrije Universiteit, Amsterdam, The Netherlands}
\affil[3]{University of Michigan}
\affil[4]{University of Salzburg}

\usepackage{thm-restate}
\usepackage{placeins}
\usepackage{bm}
\usepackage{tabularx}
\usepackage{makecell}

\usepackage{appendix}
\usepackage[appendix=inline,bibliography=common]{apxproof}

\usepackage{tikz}
\usetikzlibrary{decorations.pathmorphing}
\newlength{\LETTERheight}
\AtBeginDocument{\settoheight{\LETTERheight}{I}}

\crefname{enumi}{condition}{conditions}
\crefname{algocf}{algorithm}{algorithms}

\ifdefined\N\else\DeclareMathOperator{\N}{\mathbb{N}}\fi
\ifdefined\Z\else\DeclareMathOperator{\Z}{\mathbb{Z}}\fi
\ifdefined\R\else\DeclareMathOperator{\R}{\mathbb{R}}\fi

\DeclareMathOperator{\diam}{diam}
\newcommand{\BigO}{O}
\newcommand{\Oish}{\widetilde{\BigO}}
\newcommand{\Omegaish}{\widetilde{\Omega}}
\newcommand{\Thetaish}{\widetilde{\Theta}}

\newcommand{\chains}{\mathcal{C}}

\DeclareMathOperator{\tc}{TC}
\DeclareMathOperator{\supershortcut}{SuperShort}

\newcommand{\hopbound}{\beta}

\newcommand{\hopdist}{\texttt{hopdist}}
\newcommand{\dist}{\texttt{dist}}
\newcommand{\eps}{\varepsilon}

\newcommand{\tail}[1]{\mathrm{tail}(e)}

\newcommand{\Althofer}{Alth\"{o}fer}

\usepackage{caption}
\usepackage{subcaption}
\usepackage{mathtools}

\usepackage{booktabs}

\newcommand{\ie}[1]{(i.e.,~#1)}

\usepackage{enumitem}
\setlist[itemize]{noitemsep, nolistsep}
\setlist[enumerate,1]{label=(\arabic*)}

\newtheoremrep{appendixtheorem}[theorem]{Theorem}
\newtheoremrep{appendixlemma}[theorem]{Lemma}
\newtheoremrep{appendixcorollary}[theorem]{Corollary}

\newif\ifpolishversion

\usepackage{environ}

\ifpolishversion

  \RenewEnviron{proof}
   {\begin{OldProof} This fact is well known within competitive programming. \end{OldProof}}
\fi

\renewcommand{\subparagraph}[1]{\paragraph{#1}}

\makeatletter
\@fleqnfalse 
\makeatother

\date{April 24, 2026}

\begin{document}

\maketitle
\begin{abstract}
For many popular \emph{graph metric sparsifiers}, such as spanners, emulators, and preservers, simple and elegant greedy algorithms are known that achieve state-of-the-art or existentially optimal tradeoffs between size and quality.
The goal of this paper is to develop and analyze comparable greedy algorithms for nearby objects in \emph{graph metric augmentation}.
We show the following:
\begin{itemize}
\item A simple greedy algorithm for shortcut sets achieves the state-of-the-art size/hopbound tradeoff recently proved by Kogan and Parter (2022), up to $O(\log n)$ factors in the size.
Moreover, with an additional preprocessing step, the greedy algorithm subpolynomially \emph{improves} on the previous size bounds in some range of parameters.

\item The same greedy algorithm was already known to be existentially optimal for the size/hopbound tradeoff for hopsets, by an analysis of Berman, Raskhodnikova, and Ruan (2010) introduced for transitive-closure spanners.
We provide a completely different analysis showing that the algorithm is also existentially optimal (up to $O(\log n)$ factors) for the \emph{matching} hopset problem, in which one has a budget of roughly $O(m)$ additional edges (for an $m$-edge input graph).
\end{itemize}
\end{abstract}

\pagenumbering{gobble}
\clearpage
\tableofcontents
\clearpage

\pagenumbering{arabic}

\section{Introduction}

We study \emph{shortcut sets} and \emph{hopsets}.
These are graph-theoretic primitives which have enjoyed a recent surge of popularity due to their applications in algorithms for computing reachability, distances, or shortest paths.
Some examples include parallel algorithms \cite{UY91, cohen2000, elkin2019RNC, cao_efficient_2020, cao2020improved, cao2023exact,fineman2018nearly,JLS19}, dynamic algorithms \cite{gutenberg2020, henzinger2014, bernstein2021deterministic, LN2022}, and distributed algorithms \cite{elkin2019journal, elkin2017, censor2021}, which in turn have had applications to certain static algorithms, e.g.\ for computing flows \cite{bernstein2021deterministic, chen2022maximum}.
Formally:

\begin{definition} [Hopdistance]
The \textbf{hopdistance} between two nodes $(s, t)$ in a (possibly weighted) graph $G$ is the minimum number of edges in any $s \leadsto t$ shortest path.
\end{definition}

\begin{definition} [Shortcuts]
For an unweighted graph $G$, a \textbf{shortcut set} is a set of additional unweighted edges $H$ in the transitive closure\footnote{Recall: The transitive closure is the graph $G^*$ over the same vertex set as $G$, which has each edge $(u, v)$ iff there is a $u \leadsto v$ path in $G$.} of $G$.
The \textbf{hopbound} of $H$ is the max hopdistance in $G \cup H$ over all reachable pairs $(s, t)$.
\end{definition}

\begin{definition} [Hopsets]
For a weighted graph $G$, a \textbf{hopset} is a set of additional weighted edges $H$ in the distance closure\footnote{Recall: the distance closure is the same as the transitive closure, but each edge $(u, v) \in E(G^*)$ is assigned weight $w(u, v) := \dist_G(u, v)$.} of $G$.
The \textbf{hopbound} of $H$ is the max hopdistance in $G \cup H$ over all reachable pairs $(s, t)$.
\end{definition}

\begin{figure}[t]
\begin{center}
\begin{tikzpicture}
    \foreach \x in {0,1,2,3,4}
        \foreach \y in {0,1,2,3}
            \node[draw, circle, fill=black, inner sep=2pt] (N\x\y) at (\x,\y) {};

    \foreach \x in {0,1,2,3,4}
        \foreach \y [count=\yi from 1] in {0,1,2}
            \draw[->] (N\x\y) -- (N\x\yi);

    \foreach \y in {0,1,2,3}
        \foreach \x [count=\xi from 1] in {0,1,2,3}
            \draw[->] (N\x\y) -- (N\xi\y);
            
    \draw [red, ->] (0, 2) -- (2, 3);    
    \draw [red, ->] (0, 2) to[bend left=20] (3, 2);
    
    \draw [red, ->] (1, 1) to[bend left=20] (4, 1);
    \draw [red, ->] (1, 1) -- (3, 2);
    \draw [red, ->] (1, 1) -- (2, 3);
    
    \draw [red, ->] (2, 0) -- (4, 1);
    \draw [red, ->] (2, 0) -- (3, 2);
    \draw [red, ->] (2, 0) to[bend right=20] (2, 3);

\end{tikzpicture}
\end{center}
\caption{The red edges form a shortcut set of hopbound $\beta=5$ for the underlying graph, since once they are added to the graph, every reachable pair of nodes has a path of at most $5$ edges.  If all red edges are assigned weight $3$ (the distance between their endpoints in the underlying graph), then they form a hopset of hopbound $\beta=5$.}
\end{figure}

A central goal in the area is to develop constructions of shortcut sets and hopsets with a provably favorable tradeoff between size and hopbound.
There are two popular versions of this question:
\begin{itemize}
\item In one version, the goal is to prove a favorable tradeoff among the number of nodes $n$, the shortcut/hopset size $|H|$, and the hopbound $\hopbound$.
Recent work has closed this problem for hopsets: the best possible tradeoff is
$$|H| \le O\left(\frac{n^2}{\beta^2}\right)$$
in both settings of directed and undirected graphs \cite{BH23, VXX24}.
For shortcut sets, a recent breakthrough by Kogan and Parter \cite{KP22a} showed an improved bound of
$$|H| \le \Oish\left( \frac{n^{3/2}}{\hopbound^{3/2}} + \frac{n^2}{\hopbound^3}\right)$$
for directed input graphs,
but there is still a polynomial gap to the lower bound; closing or narrowing this is a major open question.
The progression of work on this version of the problem is summarized in \Cref{tbl:priorworksizeh}.

\item In the other version, sometimes called \emph{matching shortcuts/hopsets} \cite{Thorup95}, the shortcut/hopset size budget is $|H| \le O(m)$ (or perhaps $m^{1+o(1)}$), where $m$ is the number of edges in the input graph.
The goal is to obtain the best possible tradeoff between $n$ and $\hopbound$ within this budget.
This version arises in algorithmic applications, where the shortcut/hopset is explicitly added to the input graph as a preprocessing step, and the constraint $|H| \le O(m)$ implies that this operation doesn't asymptotically increase the size of the input graph.
It is a major open problem to obtain any upper bound in this setting that beat those known for $|H| \le O(n)$.\footnote{In particular, for matching hopsets one may assume without loss of generality that the input graph is weakly connected, and so $m \ge \Omega(n)$.  Thus any bounds for the $|H| \le O(n)$ setting immediately transfer to the matching hopset problem.}
Concretely, the current bounds are $\Omega(n^{2/9}) \le \hopbound \le \Oish(n^{1/3})$ for matching shortcuts and $\Omega(n^{1/3}) \le \hopbound \le O(n^{1/2})$ for matching hopsets, with both lower bounds proved by Hoppenworth, Xu, and Xu \cite{HXX25}.\footnote{The lower bound of $\Omega(n^{1/3})$ for matching hopsets of weighted graphs is not written in their paper, but follows directly from their methods (personal communication).}
\end{itemize}

\begin{table}[th]
\begin{center}
\begin{tabular}{lll}
\toprule
\textbf{Hopset size} & \textbf{Shortcut Set size} & \textbf{Citation}\\
\midrule
$\Oish\left(n^2 \hopbound^{-2}\right)$ & $\Oish\left(n^2 \hopbound^{-2}\right)$ & \cite{UY91} (see also \cite{Fineman19, JLS19, BGJRW12})\\
$O\left(n^2 \hopbound^{-2}\right)$ & $O\left(n^2 \hopbound^{-2}\right)$ & \cite{BRR10}\\
& $\Oish\left( n^{3/2} \hopbound^{-3/2} + n^2 \hopbound^{-3}  \right)$ & \cite{KP22a}\\
\midrule
& $\Omega\left(n^{\frac{2d}{d+1}}\beta^{-d}\right)$ & \cite{HP18} (see also \cite{Hesse03, BHT23})\\
$\Omega\left( n^{2} \hopbound^{-3} \right)$ & & \cite{KP22} (see also \cite{BHT23})\\
$\Omegaish\left( n^2 \hopbound^{-2} \right)$ & $\Omegaish\left( n^{4/3}\hopbound^{-4/3} + n^{5/4} \hopbound^{-1} \right)$ & \cite{BH23}\\
& $\Omega\left(n^{1 + \frac{d-1}{d+3}} \hopbound^{-(d-1)} \right)$ & \cite{VXX24}\\
\bottomrule
\end{tabular}
\end{center}
\caption{\label{tbl:priorworksizeh} Prior work on the achievable hopbounds for shortcut sets and exact hopsets with hopbound $\hopbound$.  All shortcut set results are for the setting of directed input graphs, while all hopset results apply in both the directed and undirected settings.  In the lower bounds, $d$ may be any positive integer.  Many of these results were initially phrased by taking hopset size $|H|$ on input and investigating the attainable hopbound $\beta$; we have inverted these expressions to facilitate comparison.}
\end{table}

\subsection{Graph Metric Sparsifiers and Greedy Algorithms \label{sec:sparsaug}}

Shortcut sets and hopsets live within the research area of graph metric augmentation, where the goal is to \emph{add} a small number of edges to a graph to improve its distance or reachability properties.
A related area is that of graph metric sparsification, where the goal is to \emph{remove} as many edges as possible from a graph while approximately preserving certain distance or reachability properties.
Some commonly studied graph metric sparsifiers include spanners \cite{PU89jacm}, emulators \cite{DHZ00}, preservers \cite{CE06, AB18}, transitive reductions \cite{AGU72}, etc.
Although sparsification and augmentation seem to be opposites, the community has recently started to view them as deeply related.
They have long drawn on similar technical toolkits, and a 2020 survey by Elkin and Neiman \cite{elkin_near-additive_2020} took some steps towards formalizing the relationships between these areas, with further formal reductions between the areas discovered recently (e.g.~\cite{KP22, BHT23, KP25}).

For many graph metric sparsifiers, there is a theorem of the following flavor: some simple greedy construction will achieve either
\begin{itemize}
\item \emph{State-of-the-art}, meaning that it produces sparsifiers that provably match the currently-known size/quality tradeoff attainable by other constructions, or
\item \emph{Existential optimality}, meaning that it produces sparsifiers that provably match the \emph{best possible} size/quality tradeoff that exists.
Existential optimality is stronger, and hence preferable when it can be attained.
Interestingly, existential optimality can often be proved even when the quantitative tradeoff is not yet known.
\end{itemize}
Some of the results developing greedy algorithms of this type include \cite{ADDJS93, BDPV18, BP19, PST24, ENS14, FS16, Knudsen14, Kavitha17, BKMP10, ABDKS20, EGN21, CE06, BHT23, GHP20}; we survey these in more depth in \Cref{sec:relatedwork}.
Besides their attractive simplicity, greedy algorithms are popular because they often have simpler analyses, they are typically deterministic, they generalize well to different graph classes \cite{FS16}, and there are some experimental results suggesting that they may produce higher-quality sparsifiers on some graphs of interest \cite{Farshi14, ABSHKS21}.

One might therefore expect that objects from graph metric \emph{augmentation} would admit similar greedy algorithms.
However, only one such greedy algorithm has been discovered (and that over a decade ago).
That is \Cref{alg:greedyhop}, which is a tweaked version of an elegant algorithm introduced by Berman, Raskhodnikova, and Ruan \cite{BRR10} to construct transitive-closure spanners (these are closely related objects to shortcut sets; see \Cref{sec:relatedwork} for details).\footnote{The algorithm of Berman, Raskhodnikova, and Ruan \cite{BRR10} computes a specific edge in each round, which provably makes some amount of ``progress'' towards a correct shortcut set (or in their case, transitive-closure spanner), rather than more naively taking any edge that maximizes progress.  They also implicitly use a slightly different potential function, discussed in \Cref{sec:brr}.}
This algorithm begins with an initially-empty shortcut/hopset, and then adds the edge in each round that reduces the sum-of-hopdistances among unsatisfied node pairs as much as possible, until a shortcut/hopset of the desired hopbound is reached.
The analysis of Berman et al.~\cite{BRR10}, which we recap in \Cref{sec:brr}, implies that this algorithm produces shortcut/hopsets of hopbound $\beta$ and size $O(n^2 / \hopbound^2)$.
This bound is optimal for hopsets in their standard version, but it is open whether it is optimal for \emph{matching} hopsets.
The bound is polynomially suboptimal for both standard and matching shortcut sets (although it is still open whether the \emph{algorithm} is optimal; that is, a better analysis of the greedy algorithm could still conceivably attain an optimal bound). 

\DontPrintSemicolon
\begin{algorithm}[t]
\SetKwFor{RepTimes}{repeat}{times}{end}
\textbf{Input:} Graph $G$, desired hopbound $\hopbound$\;~\\



Let $H \gets \emptyset$\;
\While{$H$ has hopbound $>\hopbound$}{

~

    Add to $H$ the edge $e$ from the transitive/distance closure $G^*$ that maximizes the decrease in the quantity
    \begin{align*}
    \phi(H) :=\sum \limits_{\substack{(s, t) \in E(G^*) \text{ with}\\ \hopdist_{G\cup H}(s,t)> \hopbound}} \hopdist_{G \cup H}(s, t).
    \end{align*}
}

\Return{$H$}

\caption{\label{alg:greedyhop} A greedy algorithm for shortcut sets and hopsets}
\end{algorithm}

\subsection{Our Results \label{sec:ourresults}}

The contribution of this paper is to provide new and stronger analyses of \Cref{alg:greedyhop}.
We show the following results.

\begin{theorem} \label{thm:ssnbound}
Given an $n$-node DAG on input, \Cref{alg:greedyhop} constructs a shortcut set with hopbound $\hopbound$ and size
\begin{equation*}
|H| \le \Oish\left( \frac{n^{3/2}}{\hopbound^{3/2}}+ \frac{n^2}{\hopbound^3}  \right).
\end{equation*}
Moreover, with an additional standard preprocessing step (see \Cref{sec:preprocess}), the same bound applies for general directed input graphs.
Up to hidden $O(\log n)$ factors, this ties the state-of-the-art bound previously proved by Kogan and Parter \cite{KP22a} by a different algorithm.
\end{theorem}

This is proved in \Cref{sec:shortcut}.
In addition to interest in greedy algorithms, it may be interesting to note that this is the first algorithm and analysis that avoid the use of \emph{chain covers}, an ubiquitous technical step in previous shortcut set constructions and a common computational bottleneck.
We further show that a modified version of \Cref{alg:greedyhop}, which re-introduces chain covers and has an additional preprocessing phase, \emph{improves} on the bound obtained by Kogan and Parter \cite{KP22a} by subpolynomial terms.

\begin{theorem} \label{thm:greedyimprovedss}
Given an $n$-node DAG on input, \Cref{alg:chain-greedy} (which is \Cref{alg:greedyhop} with some additional preprocessing steps) constructs a shortcut set with hopbound $O(n^{1/3})$ and size $|H| \le \BigO(n\log^{*} n)$.
\end{theorem}

This result essentially converts a $O(\log^2 n)$ factor into a $O(\log^{*} n)$ factor, but it only applies for one particular point on the size/hopbound tradeoff curve.
We give more details in \Cref{sec:chain-shortcut-greedy}.
Next, we show that \Cref{alg:greedyhop} is near-existentially optimal for matching hopsets, in a sense that we next formalize (recall that it is already known that \Cref{alg:greedyhop} is existentially optimal for standard hopsets).

\begin{definition} [$\exopt$]
We write $\exopt(n, m, h)$ to denote the smallest integer $\hopbound$ such that every $(\le n)$-node, $(\le m)$-edge graph has a hopset of size $|H| \le h$ with hopbound $\le \hopbound$.\footnote{Note that the definition of $\exopt$ depends on the graph setting \ie{it may be a different function in the setting of directed input graphs vs.\ undirected input graphs}.  Since this result applies in either setting (with the respective interpretation of $\exopt$), we will not differentiate between these settings in our notation.  However, we note that $\exopt$ may be a different function still if we were to restrict to \emph{unweighted} input graphs, and this result does \emph{not} extend to the unweighted setting.}
\end{definition}

Formally, the matching hopset problem is to determine the asymptotic value of $\exopt(n, m, m)$.\footnote{Some algorithmic applications actually care about a function of the form $\exopt(n, m, m^{1+o(1)})$.  We will use $m$ in the third parameter for concreteness and simplicity, but our results extend immediately to this setting as well.}
An algorithm would be considered \emph{existentially optimal} if, for any $n$-node, $m$-edge input graph and hopbound parameter $O(\exopt(n, m, m))$, it always returned a hopset of $O(m)$ size.
We prove that this is \emph{nearly} the case for \Cref{alg:greedyhop}:

\begin{theorem} \label{thm:hopmexopt}
Given an $n$-node, $m$-edge graph and a hopbound parameter $\hopbound := 2\cdot \exopt(n, 2m, cm/ \log n)$ for a sufficiently small constant $c$, \Cref{alg:greedyhop} returns a hopset of size $|H| \le m$.
In other words, \Cref{alg:greedyhop} is existentially optimal for matching hopsets, up to $\Oish(1)$ factors in the parameters to $\exopt$.
\end{theorem}

We prove this theorem in \Cref{sec:hopsets}.
Since this theorem statement is a bit unusual, let us take a moment to clarify.
The usual statement of (near-)existential optimality would be that, given an $n$-node, $m$-edge input graph and a hopbound parameter of the form $\hopbound = \Oish(\exopt(n, m, m))$, the algorithm returns a hopset of size $|H| \le \Oish(m)$.
The key difference in Theorem \ref{thm:hopmexopt} is that, instead of losing $O(\log n)$ factors in $|H|$, we lose a $O(\log n)$ factor in the third parameter to $\exopt$ (and we also lose a factor of $2$ in the second parameter).
It is likely, although not certain, that this implies existential optimality in the usual sense.
That is: it is expected that $\exopt$ depends at worst polynomially on its parameters; for example, perhaps $\exopt(n, m, h) = \Theta(n^a m^b h^d)$ for some constants $a, b, d$.
Assuming this is the case, we have
\begin{align*}
\exopt\left(n, 2m, \frac{cm}{\log n}\right) &= \Theta\left(n^a (2m)^b \left(\frac{cm}{\log n}\right)^d\right)\\
&= \Thetaish\left(n^a m^b m^d\right)\\
&= \Thetaish\left(\exopt(n, m, h)\right).
\end{align*}
meaning that Algorithm \ref{alg:greedyhop} is near-existentially optimal for matching hopsets in the usual sense.



\subsection{Other Related Work \label{sec:relatedwork}}

As previously mentioned, many graph metric sparsifiers have simple greedy algorithms that achieve state-of-the-art or existentially optimal tradeoffs between size and quality.
The following is a brief survey of this line of work:
\begin{itemize}
\item A \emph{multiplicative $k$-spanner} of a graph $G$ is a subgraph $H$ that preserves all pairwise shortest path distances up to a $\cdot k$ factor \cite{PU89jacm}.
In a classic 1993 paper by \Althofer{}, Das, Dobkin, Joseph, and Soares \cite{ADDJS93}, the authors studied a simple incremental greedy construction algorithm, and proved that it achieves an existentially optimal tradeoff among the number of edges $|E(H)|$, the number of nodes in the input graph $n$, and the stretch $k$.
This proof technique has since been extended to show existential optimality for the analogous greedy algorithm for fault-tolerant spanners \cite{BDPV18, BP19, PST24}, light spanners \cite{ENS14}, and spanners in various graph classes \cite{FS16}.

\item An \emph{additive $k$-spanner} of a (typically unweighted) graph $G$ is a subgraph $H$ that preserves all pairwise shortest path distances up to a $+k$ factor.
Knudsen \cite{Knudsen14} adapted a previous algorithm by Baswana, Kavitha, Mehlhorn, and Pettie \cite{BKMP10} into a greedy algorithm that matches the state-of-the-art tradeoff among the number of edges $|E(H)|$, the number of nodes in the input graph $n$, and the additive error $k$.
This proof technique has since been extended to show state-of-the-art greedy algorithms for pairwise spanners \cite{Kavitha17} and for a related notion of additive spanners for weighted input graphs \cite{ABDKS20, EGN21}.

\item A \emph{distance preserver} of a graph $G$ and a set of demand pairs $P = V(G) \times V(G)$ is a subgraph $H$ that exactly preserves the distance among all pairs in $P$.
Coppersmith and Elkin \cite{CE06} proved that an incremental greedy algorithm, essentially adding a well-chosen shortest path for some demand pair in each round, matches the state-of-the-art tradeoff among size $|E(H)|$, number of nodes $n$, and number of demand pairs $|P|$.
Bodwin, Hoppenworth, and Trabelsi \cite{BHT23} proved that a different decremental greedy algorithm achieves existential optimality for distance preservers, as well as for reachability preservers (which more weakly preserve only reachability among demand pairs).
This proof technique has also been applied to achieve existential optimality for distance/reachability-preserving minors \cite{GHP20}. 
\end{itemize}

We next discuss prior work on shortcut set algorithms with an emphasis on \emph{running time}.
A common technical step in the construction of shortcut sets is the construction of a \emph{chain cover}.
It turns out that we may assume without loss of generality that the input graph is a DAG, thus the goal is then to compute a maximal set of vertex-disjoint paths in the transitive closure of length $\Omega(\hopbound)$ each.
In \cite{caceres_minimum_2023}, C{\'a}ceres gave an almost-linear time algorithm to a suitable chain cover.
In followup work \cite{caceres_maximum_2025}, C{\'a}ceres, Grigorjew, Jiamjitrak, and Tomescu developed greedy algorithms for this problem.
These algorithms imply near-linear running time for shortcut sets with size $\Oish(n)$ and suboptimal hopbound $O(\sqrt{n})$ (see also \cite{cao_efficient_2020}).
Since hopsets and shortcut sets are often applied in parallel settings, there is also considerable literature on work-efficient construction algorithms.
See e.g.~\cite{Fineman19, JLS19, BGJRW12} and references within.
The algorithms used in this setting are typically based on classical random sampling approaches, and are not greedy.
There have also been recent hardness-of-approximation results for the problem of computing the best shortcut/hopset for a particular input instance \cite{dinitz_approximation_2025, chalermsook2025shortcuts}.

Related to shortcut sets and hopsets, another commonly studied object in graph metric augmentation is an \emph{$\alpha$-approximate hopset}.
While a shortcut set guarantees that every reachable pair has \emph{some} path with $\le \beta$ hops, and a hopset more specifically guarantees a \emph{shortest} path with $\le \beta$ hops, an approximate hopset guarantees a \emph{$\alpha$-approximate shortest path} with $\le \beta$ hops.
Following bounds by Kogan and Parter \cite{KP22a}, Berenstein and Wein~\cite{BW23} obtained bounds for $(1+\eps)$-approximate hopsets in directed graphs with quasipolynomially-bounded edge weights that essentially match the known bounds for shortcut sets.
Later, Hoppenworth, Xu, and Xu \cite{HXX25} showed that when edge weights can be unbounded, the attainable bounds for $(1+\eps)$-approximate hopsets are no better than those attainable for exact hopsets.
There has also been considerable work on approximate hopsets in the setting of undirected graphs \cite{HP18, ABP17}.
Our techniques do not seem to extend straightforwardly to analyze greedily algorithms for approximate hopsets, but attaining results along these lines is an interesting open problem.

Another related object that has received attention from the community is \emph{transitive-closure spanners}; see e.g.~\cite{BGJRW12, Raskhodnikova10}.
Transitive-closure spanners are strictly stronger objects than shortcut sets; they essentially require $\dist_H(s, t) \le \hopbound$ rather than $\dist_{G \cup H}(s, t) \le \hopbound$.
This strengthening generally prevents one from proving extremal tradeoffs between $\beta$ and the size of $H$ like those available for shortcut/hopsets, because (for example) a transitive-closure spanner may not remove any edges from a directed biclique input graph.
Nonetheless, the similarity in problem statements allows some results and techniques to transfer between them \cite{BRR10}.

\section{Warmup: Simple Greedy Achieves Near-Optimal Hopsets}
\label{sec:warmup}
We will start by describing an analysis of \Cref{alg:greedyhop} for hopsets, which matches the bound achieved by the classical folklore sampling algorithm, and in some sense can be viewed as a derandomization of the algorithm.
This theorem follows from a similar analysis of \cite{BRR10}, but with minor adaptations which we will recap in this section. In particular, we adapt their algorithm and analysis to a version in which the edges are added based on a slightly different potential function. In \Cref{sec:brr}, we briefly explain the difference in their analysis.

\begin{theorem} \label{thm:opthopsets}
When we run \Cref{alg:greedyhop} for hopsets on an $n$-node input graph with hopbound parameter $\hopbound$, it returns a hopset of size $\Oish(n^2 / \hopbound^2)$.
\end{theorem}

The \emph{statement} of Theorem \ref{thm:opthopsets} is not a new contribution of this paper; rather, our goal in proving this theorem is to explain the current knowledge from prior work \cite{BRR10} in more technical depth, and also to give a warmup proof to introduce some of our main technical ideas in a simpler way.

\subsection{Potential Reduction Lemma and Proof of \Cref{thm:opthopsets}}

The quantity $\phi(H)$ written in the algorithm will be called the \emph{potential} of the partially-constructed shortcut set $H$, and we write
\begin{equation*}
\Delta(u, v \mid H) \coloneqq  \phi(H) - \phi(H \cup \{(u, v)\})
\end{equation*}
for the amount that a new edge $(u, v)$ would reduce potential when added to the current hopset $H$.
The following key lemma guarantees the existence of an edge that significantly reduces potential:
\begin{lemma} [Potential Reduction Lemma] \label{lem:opthspotential}
In each round of the algorithm, there is a pair of nodes $(u, v)$ for which
$$\Delta(u, v \mid H) \ge \phi(H) \cdot \Omega\left(\frac{\beta^2}{n^2}\right).$$
\end{lemma}
We will prove this lemma in the following subsection; for now, let us see quickly how it implies \Cref{thm:opthopsets}.
Initially, potential is $\phi(H) \le n^3$, since there are at most $n^2$ node pairs that contribute at most $n$ to potential each.
The potential reduction lemma implies that potential will reduce by a constant fraction every $n^2 / \hopbound^2$ rounds of the algorithm.
Since potential must always be a nonnegative integer, this implies that potential reaches $0$ within $O(n^2 \log n / \hopbound^2)$ rounds.
Once potential is $0$, this implies that no more node pairs have hopbound $>\hopbound$, and so the algorithm halts, having added one edge to the hopset $H$ per round and thus yielding the desired hopset size.

\paragraph{Proof of \Cref{lem:opthspotential}.}

Choose a node pair $(u, v)$ uniformly at random among the reachable pairs $(u, v)$; our strategy is to bound the expectation of $\Delta(u, v \mid H)$.
Consider any node pair $(s, t)$ with $\hopdist_{G \cup H}(s, t) > \beta$, and let $\pi(s, t)$ be a hop-minimal shortest path for this pair.
If we sample $(u, v)$ such that $u$ is one of the first third of the nodes along $\pi(s, t)$, and $v$ is one of the last third of the nodes along $\pi(s, t)$, then by adding the hopedge $(u, v)$ we will reduce $\hopdist_{G \cup H}(s, t)$ by a constant factor.
Thus, letting $G^*$ be the transitive closure of $G$, we have:
\begin{align*}
\mathbb{E}\left[\Delta(u, v \mid H) \right] &\ge \sum \limits_{\substack{(s, t) \in E(G^*) \text{ with}\\ \hopdist_{G\cup H}(s,t)> \hopbound}} \Pr\left[ u \text{ sampled in first }, v \text{ in last third of } \pi(s, t)\right] \cdot \Theta\left( \hopdist_{G \cup H}(s, t)\right)\\
&= \sum \limits_{\substack{(s, t) \in E(G^*) \text{ with}\\ \hopdist_{G\cup H}(s,t)> \hopbound}} \Theta\left( \frac{\hopdist_{G\cup H}(s, t)}{n} \right)^2 \cdot \Theta\left( \hopdist_{G \cup H}(s, t)\right)\\
&\ge \Theta\left( \frac{\hopbound}{n} \right)^2 \cdot \sum \limits_{\substack{(s, t) \in E(G^*) \text{ with}\\ \hopdist_{G\cup H}(s,t)> \hopbound}} \Theta\left( \hopdist_{G \cup H}(s, t)\right)\\
&= \Theta\left( \frac{\hopbound}{n} \right)^2 \cdot \phi(H).
\end{align*}

There exists a possible choice of $(u, v)$ that matches or exceeds this expectation, and this choice of $(u, v)$ therefore satisfies the potential reduction lemma.

As this analysis bounds the greedy potential reduction by the expected value from a uniformly random node pair, it can be seen as a formal connection between the folklore sampling algorithm and the greedy approach.

\subsection{Analysis of Berman, Raskhodnikova, and Ruan \cite{BRR10} \label{sec:brr}}

Berman, Raskhodnikova, and Ruan improved the previous bounds by removing the $\log$ factors; their approach may be explained as follows.
First, they use a tweaked version of \Cref{alg:greedyhop}, which can be seen as follows: the sum defining $\phi(H)$ is over demand pairs with hopdistance $> \hopbound/2$ (rather than $> \hopbound$).
With this change, consider a round of the algorithm, and let $p$ be the number of demand pairs with hopdistance $> \hopbound$.
We know that $p > 1$, or else the algorithm halts.
So by averaging, there exists a pair of nodes $(s, t)$ with
$$\hopdist_{G \cup H}(s, t) \ge \frac{\phi(H)}{p}.$$
Let $\pi(s, t)$ be a hop-minimal shortest path for this pair, let $u$ be the $(|\pi(s, t)|/4)$-th node from the beginning of $\pi(s, t)$, and let $v$ be the $(|\pi(s, t)|/4)$-th node from the end.
Notice that every node pair from the set
$$\bigg(\pi(s, t)[s \leadsto u] \bigg) \times \bigg(\pi(s, t)[v \leadsto t]\bigg)$$
has hopdistance $>\hopbound/2$ (and hence contributes to potential), and will have its hopdistance reduced by $\Theta(\hopdist_{G \cup H}(s, t))$ if we add $(u, v)$ as a hopedge.
Thus, $(u, v)$ will reduce potential by at least
\begin{align*}
& \Theta\left( \hopdist_{G \cup H}(s, t)\right)^3\\
\ge \ & \Theta\left(\frac{\phi(H)}{p} \cdot \hopdist_{G \cup H}(s, t)^2\right)\\
\ge \ & \phi(H) \cdot \Theta\left(\frac{\hopbound}{n}\right)^2 \tag*{since $p \le n^2$ and $\hopbound_{G \cup H}(s, t) > \hopbound$.}
\end{align*}

This leads to an alternate proof of the potential reduction lemma, which implies that the algorithm halts within $\Oish(n^2 / \hopbound^2)$ rounds, as before.
In order to remove the log factors from this bound, they point out that by similar logic, $p$ must decrease by $\Theta(\hopbound)^2$ in each round.
Since initially $p \le n^2$, this leads to a bound of $O(n^2 / \hopbound^2)$.

\section{Simple Greedy Achieves State-of-the-Art Shortcut Sets}
\label{sec:shortcut}
In this section we will prove \Cref{thm:ssnbound}, analyzing the shortcut sets returned by \Cref{alg:greedyhop}.
The claim that the output shortcut set has hopbound $\hopbound$ is immediate from the algorithm, so our remaining goal will be to show that the algorithm halts within the claimed number of rounds.
We will assume where convenient that the hopbound parameter $\beta$ is at least a sufficiently large constant, since for $\beta \le O(1)$ the claimed size bound is $|H| \le \Oish(n^2)$, which is trivial.

\subsection{Potential Reduction Lemma and Proof of \Cref{thm:ssnbound}}

As in the warmup proof, we will write $\phi(H)$ for the potential (as written in the algorithm), and $\Delta(u, v \mid H)$ for the amount adding the edge $(u, v)$ to the hopset $H$ would reduce potential.
The main workhorse is the following lemma establishing that each round of the algorithm significantly reduces the potential:

\begin{lemma}[Potential Reduction Lemma] \label{lem:sspotential}
In each round of the algorithm, for any parameter $\sigma$ that is at least a sufficiently large constant, there is a pair of nodes $(u, v)$ for which
\begin{equation*}
\Delta(u, v \mid H) \ge \phi(H) \cdot \Omega\left( \min\left\{ \frac{\hopbound^3}{\sigma n^2}, \frac{\sigma}{n} \right\}\right).
\end{equation*}
\end{lemma}

We will prove \Cref{lem:sspotential} in the following subsection; for now, let us see how it implies \Cref{thm:ssnbound}.
We calculate the balance point of the two terms in \Cref{lem:sspotential} as follows:
\begin{align*}
\frac{\sigma}{n} &= \frac{\hopbound^3}{\sigma n^2}\\
\sigma^2 &= \hopbound^3 n^{-1}\\
\sigma &= \hopbound^{3/2} n^{-1/2},
\end{align*}
and under this setting of $\sigma$, the terms balance to give
\begin{align*}
\Delta((u, v) \mid H) &\ge \Omega\left( \frac{\hopbound^{3/2}}{n^{3/2}} \right).
\end{align*}
However, this setting $\sigma=\beta^{3/2} n^{-1/2}$ is valid only when this expression is at least a sufficiently large constant.
If not, then we are in the parameter regime where $\beta \le O(n^{1/3})$.
We may then instead set $\sigma$ to be a sufficiently large constant, giving 
\begin{align*}
\Delta(u, v \mid H) &\ge \phi(H) \cdot \Omega\left( \frac{\hopbound^{3}}{n^{2}} \right)
\end{align*}
(since this term is $O(1/n)$, so the second term in the $\min$ is ignored).
So combining the two cases, potential reduces in each round by at least
\begin{align*}
\Delta(u, v \mid H) &\ge \phi(H) \cdot \Omega\left( \min\left\{ \frac{\hopbound^{3/2}}{n^{3/2}},  \frac{\hopbound^{3}}{n^{2}} \right\}\right).
\end{align*}
Thus, the potential reduces by a constant factor every
$$\max\left\{ \frac{n^{3/2}}{\hopbound^{3/2}},  \frac{n^{2}}{\hopbound^{3}} \right\}$$
rounds.
Since potential must always be a nonnegative integer, this implies that potential reaches $0$ within
$$\Oish\left(\max\left\{ \frac{n^{3/2}}{\hopbound^{3/2}},  \frac{n^{2}}{\hopbound^{3}} \right\} \right)$$
rounds, and the algorithm halts.


\subsection{Proof of \Cref{lem:sspotential}}
\label{sec:potential-reduction-lemma-proof}

We define
\begin{equation*}
A(H) \coloneqq  \left\{ (s, t) \text{ in transitive closure of } G \ \mid \ \dist_{G \cup H}(s, t) > \hopbound\right\},
\end{equation*}
as the node pairs that contribute positively to the sum defining $\phi(H)$.
For any nodes $x, y$, we denote by $\pi_{G \cup H}(x, y)$ a canonical choice of $u \leadsto v$ shortest path in $G \cup H$.
We assume that these shortest paths are selected to be \emph{consistent}: that is, if $u, v \in \pi_{G \cup H}(x, y)$ with $u$ preceding $v$, then $\pi_{G \cup H}(u, v) \subseteq \pi_{G \cup H}(x, y)$.
This is a common technical property; it can be obtained, for example, by temporarily adding some small random noise to the edge weights and choosing the weighted shortest path.
We let $\Pi_A$ be the set of all canonical shortest paths among pairs in $A(H)$.

We denote by $|\pi|$ the number of nodes in the path $\pi$, and $\len(\pi)$ the number of edges in the path $\pi$ (so $|\pi| = \len(\pi)+1$).
For a path $\pi \in \Pi_A$, we will say that its \emph{prefix} $\texttt{prefix}(\pi)$ is the subpath consisting of the first $\lfloor |\pi|/4 \rfloor$ nodes, and its \emph{suffix} $\texttt{suffix}(\pi)$ is the subpath consisting of the last $\lfloor |\pi|/4 \rfloor$ nodes.
We write $\prefdeg(v), \suffdeg(v)$ respectively for the number of prefixes and suffixes of paths in $\Pi_A$ that contain a node $v$.


\setlength{\lineskip}{0pt}
\setlength{\lineskiplimit}{0pt}
\begin{lemma} \label{lem:basepath}
In each round, there exists a canonical shortest path $\pi^*$ in $G$ with $|\pi^*| \le \frac{\hopbound}{4}$ nodes, and
$$\sum \limits_{u \in \pi^*} \suffdeg(v) \ge \Omega\left( \frac{\hopbound \cdot \phi(H)}{n} \right).$$
\end{lemma}
\begin{proof}
Sample a path $\pi \in \Pi_A$ with probability proportional to its number of nodes; that is,
$$\Pr[\text{select } \pi \in \Pi_A] \coloneqq  \frac{|\pi|}{\sum \limits_{\pi \in \Pi_A} |\pi|} = \Theta\left(\frac{|\pi|}{\phi(H)}\right).$$
Then, choose $\pi^*$ to be a contiguous subpath of $\texttt{suffix}(\pi)$, obtained as follows.
With probability $\frac{1}{2}$, pick $\pi^*$ as the first $\frac{\hopbound}{4}$ nodes of $\texttt{suffix}(\pi)$.
Otherwise, with probability $\frac{1}{2}$, pick $\pi^*$ by choosing a uniform-random node along $\texttt{suffix}(\pi)$, and then taking the following $\frac{\hopbound}{4}$ nodes (truncating the path early if the end of $\texttt{suffix}(\pi)$ is reached).
The purpose of this process is that, for a node $u \in \pi$, the probability that we have $u \in \pi^*$ is at least $\Omega(\hopbound / |\pi|)$.
This enables the following calculation:

\begin{align*}
\mathbb{E}\left[ \sum \limits_{u \in \pi^*} \suffdeg(u)\right] &\ge \sum \limits_{\pi \in \Pi_A} \Pr[\pi \text{ sampled}] \cdot \sum \limits_{u \in \texttt{suffix}(\pi)} \Pr[u \in \pi^*] \cdot \suffdeg(u)\\
&\ge \Theta\left( \sum \limits_{\pi \in \Pi_A} \frac{|\pi|}{\phi(H)} \cdot \sum \limits_{u \in \pi} \frac{\hopbound}{|\pi|} \cdot \suffdeg(u) \right)\\
&=\Theta\left( \frac{\hopbound}{\phi(H)} \cdot \sum \limits_{\pi \in \Pi_A} \sum \limits_{u \in \texttt{suffix}(\pi)} \suffdeg(u) \right)\\
&= \Theta\left( \frac{\hopbound}{\phi(H)} \cdot \sum \limits_{u \in V} \suffdeg(u)^2 \right)\\
&= \Theta\left( \frac{\hopbound}{n \cdot \phi(H)}\right) \cdot \left(\sum \limits_{u \in V} \suffdeg(u) \right)^2 \tag*{Cauchy-Schwarz}\\
&= \Theta\left( \frac{\hopbound}{n \cdot \phi(H)}\right) \cdot \left(\frac{\phi(H)}{4} \right)^2\\
&= \Theta\left( \frac{\hopbound \cdot \phi(H)}{n}\right). \tag*{\qedhere}
\end{align*}
\end{proof}

In the following, we fix $\pi^*$ as some path in $G$ satisfying the previous lemma.
Next, let $\sigma$ be a parameter that is at least a sufficiently large constant, and let us say that a path $\pi \in \Pi_A$ \emph{heavily intersects} $\pi^*$ if $\left|\texttt{suffix}(\pi) \cap \pi^*\right| \ge \sigma$.
The next two lemmas functionally split into two cases, by whether or not a constant fraction of the sum in \Cref{lem:basepath} is contributed by paths that heavily intersect $\pi^*$.
Note that \Cref{lem:basepath} implies that the premises of one lemma or the other must hold.

\begin{lemma} \label{lem:ssheavyint}
Suppose that
$$\sum \limits_{\substack{\pi \in \Pi_A,\\\pi \text{ heavily intersects } \pi^*}} \left| \texttt{suffix}(\pi) \cap \pi^* \right| \ge \Omega\left( \frac{\hopbound \cdot \phi(H)}{n} \right).$$
Then there exists a pair of nodes $(u, v)$ with $\Delta((u, v) \mid H) \ge \phi(H) \cdot \Omega\left( \frac{\sigma}{n} \right)$.
\end{lemma}
\begin{proof}
Uniformly sample a pair of nodes $(u, v)$ that are both on $\pi^*$, with $u$ exactly $\sigma/2$ steps before $v$.
Notice that, if we add $(u, v)$ to $H$, then we decrease the associated hopdistance from $u$ to $v$ from $\sigma/2$ to $1$, which is $\Theta(\sigma)$ (since we assume $\sigma$ is a sufficiently large constant).
Hence, potential will decrease by at least $\Theta(\sigma)$ times the number of paths that contain both $u$ and $v$.
Our next task is to lower bound the expected number of such paths.

For each path $\pi \in \Pi_A$ that heavily intersects $\pi^*$, we have $u, v \in \pi$ iff we choose $u \in \pi$ but $u$ is not within the last $\frac{\sigma}{2}$ nodes of $\pi$.\footnote{This step is where we use \emph{consistency} of our canonical shortest paths: since $\pi, \pi^*$ are both canonical shortest paths, their intersection is a contiguous subpath of each that contains at least $\sigma$ nodes.}
Using this, we may bound as follows.
\begin{align*}
\mathbb{E}\left[\left| \left\{ \pi \in \Pi_A \ \mid \ u, v \in \pi\right\}\right|\right] &\ge \sum \limits_{\substack{\pi \in \Pi_A,\\\pi \text{ heavily intersects } \pi^*}} \frac{|\texttt{suffix}(\pi) \cap \pi^*| - \frac{\sigma}{2}}{|\pi^*| - \frac{\sigma}{2}}\\
&\ge \sum \limits_{\substack{\pi \in \Pi_A,\\\pi \text{ heavily intersects } \pi^*}} \Omega\left(\frac{|\texttt{suffix}(\pi) \cap \pi^*|}{|\pi^*| - \frac{\sigma}{2}}\right) \tag*{since $|\texttt{suffix}(\pi) \cap \pi^*| \ge \sigma$}\\
&\ge \sum \limits_{\substack{\pi \in \Pi_A,\\\pi \text{ heavily intersects } \pi^*}} \Omega\left(\frac{|\texttt{suffix}(\pi) \cap \pi^*|}{\hopbound}\right) \tag*{since $|\pi^*| \ge |\texttt{suffix}(\pi) \cap \pi^*| \ge \sigma$}\\
&\ge \Omega\left( \frac{\phi(H)}{n}\right)
\end{align*}
where the last step follows from the hypothesis of the lemma.
Multiplying this by $\Theta(\sigma)$ as discussed previously obtains the claimed bound.
\end{proof}

\begin{lemma}
Suppose that
$$\sum \limits_{\substack{\pi \in \Pi_A,\\\pi \text{ does \textbf{not} heavily intersect } \pi^*}} \left| \texttt{suffix}(\pi) \cap \pi^* \right| \ge \Omega\left( \frac{\hopbound \cdot \phi(H)}{n} \right).$$
Then there exists a pair of nodes $(u, v)$ with $\Delta((u, v) \mid H) \ge \phi(H) \cdot \Omega\left( \frac{\beta^3}{\sigma n^2} \right)$.
\end{lemma}
\begin{proof}
Let $Q$ be the set of paths whose suffix intersects $\pi^*$, but which do not heavily intersect $\pi^*$.
Each path in $Q$ contributes $O(\sigma)$ to the sum in \Cref{lem:basepath}.
Therefore, we have
$$|Q| \ge \Omega\left(\frac{\beta \cdot \phi(H)}{\sigma n}\right).$$
Next, sample a node $u \in V$ uniformly at random, and let $Q_u \subseteq Q$ be the subset of paths from $Q$ that contain $u$ in their prefix.
Since the prefix of each path contains $\Omega(\hopbound)$ nodes, we have
\begin{align*}
\mathbb{E}[|Q_u|] &\ge |Q| \cdot \Omega\left( \frac{\hopbound}{n}\right)\\
&\ge \Omega\left( \frac{\hopbound^2 \cdot \phi(H)}{\sigma n^2}\right).
\end{align*}
Now, let $v$ be the first node along $\pi^*$ for which there exists a path in $Q_u$ with $v$ in its suffix.
Consider any path $q \in Q_u$, and let $x$ be some node in $\texttt{suffix}(q) \cap \pi^*$.
After adding the edge $(u, v)$ to the hopset, we may replace $q$ with a new path $q'$, by replacing its $u \leadsto x$ subpath with the edge $(u, v)$, followed by the $v \leadsto x$ subpath along $\pi^*$.
This reduces hopbound by:
\begin{align*}
|q| - |q'| &= \left(\len(q[u \leadsto x]) \right) - \left(1 + \len(\pi^*[v \leadsto x]) \right)\\
&\ge \left(\len(q) - \frac{\len(q)}{4} - \frac{\len(q)}{4} \right) - \left(1 + \len(\pi^*[v \leadsto x]) \right)\tag*{since $u \in \texttt{prefix}(q)$ and $x \in \texttt{suffix}(q)$}\\
&\ge \frac{\len(q)}{2} - \frac{\hopbound}{4} \tag*{since $1+\len(\pi^*) = |\pi^*| \le \frac{\beta}{4}$}\\
&= \Omega(\hopbound) \tag*{since $\len(q) \ge \beta$.}
\end{align*}
Thus, adding $(u, v)$ to $H$ reduces potential by at least $|Q_u| \cdot \Omega(\hopbound)$, giving the claimed bound.
\end{proof}



\subsection{Reducing to the Setting of DAGs \label{sec:preprocess}}

The previous analysis holds only when the input graph is a DAG.
The specific part that uses acyclicness is \Cref{lem:ssheavyint}, which relies on the assumption that the paths $\pi$ intersecting $\pi^*$ do so on a \emph{contiguous} subpath of $\pi^*$.
For a general directed graphs, it is possible that $\pi \cap \pi^*$ contains multiple disconnected subpaths of $\pi^*$, so long as $\pi$ and $\pi^*$ use those subpaths in an opposite order.

However, as previously mentioned, there is a standard preprocessing step that reduces the input graph to the setting of DAGs.
We recap it here, to keep our paper self-contained:

\begin{itemize}
\item First, the following reduction appears e.g.\ in~\cite{Raskhodnikova10}, Lemma 3.2 (but is likely folklore).
Letting $G$ be the input graph, for each strongly connected component $C$, choose a representative node $v \in V$ and add shortcut edges $\{v\} \times C \cup C \times \{v\}$.
Then contract each strongly connected component, giving a DAG $G'$.
A shortcut set on $G'$ can be straightforwardly mapped back to a correct shortcut set on $G$.

\item The previous reduction costs up to $O(n)$ edges, and thus is not directly suitable in the parameter regime where the shortcut set size budget is $|H| =: h\ll n$.
Here, we can first subsample a ``kernel'' of $h$ nodes, perform the previous reduction to spend $O(h)$ edges to reduce this kernel to a $(\le h)$-node DAG $G'$, and then map the shortcut set on $G$ back to a correct shortcut set for $G$.
More specifically, if the shortcut set for $G'$ has hopbound $\hopbound$ then the shortcut set for $G$ will be $\hopbound \cdot \Oish(n/h)$.
However, this is still enough to recover our claimed size bounds.
See e.g.~\cite{BH23}, Lemma 13, for formal analysis of this process.
\end{itemize}

Note that these are both \emph{algorithmic} reductions, and so as mentioned previously, they would need to be added to \Cref{alg:greedyhop} to achieve the claimed size bounds for general directed graphs.


\section{Simple Greedy Achieves Existentially Near-Optimal Hopsets}
\label{sec:hopsets}

In this section we will prove \Cref{thm:hopmexopt}, analyzing the hopsets returned by \Cref{alg:greedyhop}.
Again, correctness of the output hopset is immediate from the algorithm, and so we will focus on bounding its size.

\subsection{Potential Reduction Lemma and Proof of \Cref{thm:hopmexopt}}

We will analyze a run of \Cref{alg:greedyhop} with an $n$-node, $m$-edge graph $G$ and hopbound parameter $\beta \coloneqq  2\cdot \exopt(n, 2m, cm/\log n)$.
Relative to these parameter settings, we define $\phi(H)$ as the potential used in the algorithm, and $\Delta(u, v \mid H)$ as the reduction in potential when $(u, v)$ is added to the hopset $H$, as in the previous section.
The analogous potential reduction lemma for hopsets is:

\begin{lemma} [Potential Reduction Lemma] \label{lem:hspotential}
In each of the first $m$ rounds of \Cref{alg:greedyhop} for hopsets, there is a pair of nodes $(u, v)$ for which
\begin{equation*}
\Delta(u, v \mid H) \ge \phi(H) \cdot \Omega\left(\frac{\log n}{cm}\right).
\end{equation*}
\end{lemma}

We will prove this lemma in the next section; here, we quickly point out how it implies \Cref{thm:hopmexopt}.
This lemma implies that potential decreases by a constant factor every $O(cm / \log n)$ rounds.
Thus, within the first $m$ rounds, potential will decrease by a constant factor $O(\log n / c)$ times.
By choice of sufficiently small $c$, and the fact that potential must always be a nonnegative integer, this implies that potential will become $0$ within the first $m$ rounds of \Cref{alg:greedyhop}, at which point the algorithm halts.

\subsection{Proof of \Cref{lem:hspotential}}
\label{sec:hspotential-proof}

We will prove \Cref{lem:hspotential} by analyzing a different graph $G'$, derived from the input graph $G$ and the current hopset $H$.
The following claim summarizes the relevant properties of $G'$.
\begin{lemma} [New Graph $G'$] \label{lem:gprime}
In each round of the algorithm, there exists a graph $G' = (V, E \cup H, w')$ with the following properties:
\begin{itemize}
\item For all reachable pairs of nodes $(s, t)$, $G'$ has a unique shortest $s \leadsto t$ path $\pi'(s, t)$, and
\item this path $\pi'(s, t)$ is one of the shortest $s \leadsto t$ paths in $G \cup H$ that uses $\hopdist_{G \cup H}(s, t)$ edges.
\end{itemize}
\end{lemma}
\begin{proof}
We construct $G'$ as follows.
The edge set is $E \cup H$, and initially the weights are $w' \coloneqq  w$.
The first change is to fix some sufficiently small $\eps > 0$, and then add $\eps$ to every edge weight in $G'$.
This penalizes the length of a $h$-hop path by adding $\eps h$ to its length.
Thus, at this stage, for each pair $(s, t)$ the shortest $s \leadsto t$ path(s) in $G'$ are exactly the \emph{minimum-hop} shortest $s \leadsto t$ path(s) in $G \cup H$.

The second change is to fix some sufficiently smaller $0 < \delta \ll \eps$, and perturb each edge weight by an independent uniform random variable in the range $[0, \delta]$.
This breaks shortest path ties with probability $1$, and by choice of sufficiently small $\delta$, no non-shortest path can become shortest.
So there is now a unique shortest path for each demand pair, which is one of the hop-minimizing shortest paths in $G \cup H$.
This final reweighted graph is $G'$.
\end{proof}




We next prove a technical lemma.
Although the notation is a bit heavy, it essentially states the following intuitive fact: in $G'$, the reduction in sum-of-hopdistances among active pairs that can be achieved by $h$ new hopedges is at most $h$ times the reduction that can be achieved by $1$ new hopedge.
It may be more unintuitive that this does \emph{not} hold for all graphs; rather, it holds for all graphs with unique shortest paths, which is part of our reason for constructing the graph $G'$ in the first place.

 \begin{lemma}
Let $\phi', \Delta', A'$ respectively denote potential, change in potential, and active pairs (as in the previous section), with respect to the graph $G'$ and with respect to the empty hopset $H=\emptyset$.
Any hopset $H'$ for $G'$ of size $|H'|$ has hopbound at least
\begin{equation*}
\beta' \ge \frac{\phi' - |H'| \cdot \left( \max \limits_{(u, v)} \Delta'(u, v) \right)}{|A'|}.
\end{equation*}
\end{lemma}
\begin{proof}
Let $H' = \left\{ (u_1, v_1), \dots, (u_{|H'|}, v_{|H'|}) \right\}$.
For any $i$, let $H_i' \subseteq H$ denote the subset containing the first $i$ of these edges, $\left\{ (u_1, v_1), \dots, (u_i, v_i) \right\}$.
We can imagine adding the edges to $H'$, one at a time and in order, with each edge causing a reduction in the sum-of-hopdistances among pairs in $A'$.
This perspective corresponds to the identity
\begin{equation*}
\sum \limits_{(s, t) \in A'} \hopdist_{G \cup H'}(s, t)  = \phi' - \sum \limits_{i=1}^{|H'|} \sum \limits_{(s, t) \in A'} \left(\hopdist_{G \cup H'_i}(s, t) - \hopdist_{G \cup H'_{i-1}}(s, t)\right).
\end{equation*}
Next, we observe that when each edge $(u_i, v_i)$ is added to the partial hopset $H'_{i-1}$, it can possibly reduce hopdistance for an active pair $(s, t) \in A_{G'}(H'_{i-1})$ only if $u_i, v_i$ both lie on the unique shortest $s \leadsto t$ path in $G'$, and (if so) it reduces the quantity $\hopdist_{G' \cup H'_{i-1}}(s, t)$ by at most $\hopdist_{G'}(u_i, v_i) - 1$.
Meanwhile, if the same edge $(u_i, v_i)$ were added to the empty hopset $\emptyset$, then it would necessarily reduce the hopdistance for $(s, t)$ by exactly $\hopdist_{G'}(u_i, v_i) - 1$.
Thus, we may continue from our previous identity:
\begin{align*}
\sum \limits_{(s, t) \in A'} \hopdist_{G \cup H'}(s, t) &\ge \phi' - \sum \limits_{i=1}^{|H'|} \sum \limits_{(s, t) \in A'} \left( \hopdist_{G'}(u_i, v_i) - 1\right)\\
&= \phi' - \sum \limits_{i=1}^{|H'|} \sum \limits_{(s, t) \in A'} \left(\hopdist_{G \cup \{(u_i, v_i)\}}(s, t) - \hopdist_{G}(s, t)\right)\\
&\ge \phi' - \sum \limits_{i=1}^{|H'|} \Delta'(u_i, v_i)\\
&\ge \phi' - |H'| \cdot \left( \max \limits_{(u, v)} \Delta'(u, v)\right).
\end{align*}
Thus, the average hopdistance in the graph $G \cup H'$, over the pairs in $A'$, is at least
\begin{align*}
\frac{\sum \limits_{(s, t) \in A'} \hopdist_{G \cup H'}(s, t)}{|A'|} \ge \frac{\phi' - |H'| \cdot \left( \max \limits_{(u, v)} \Delta'(u, v)\right)}{|A'|}.
\end{align*}
Since there is at least one pair in $A'$ whose hopdistance in $G' \cup H'$ matches or exceeds the average, this is thus a lower bound on the hopbound of $H'$ in $G'$.
\end{proof}

The role of $G'$ in the rest of the proof is as a stepping stone to relate the potential reduction that can be achieved in $G$ to the value of the function $\exopt$.
Let $h'$ be any positive integer.
Recall that $G'$ has $n$ nodes and $m+|H|$ edges, and so by definition of $\exopt$, it has a hopset of size $|H'|=h'$ and hopbound $\exopt(n, m+|H|, h')$.
Plugging this into the bound from the previous lemma, we get
\begin{equation*}
\exopt(n, m+|H|, h') \ge \frac{\phi' - h' \cdot \left( \max \limits_{(u, v)} \Delta'(u, v) \right)}{|A'|}.
\end{equation*}

Next, let us observe the relationships between $\phi', \Delta', A'$ and the corresponding objects on $G$.
The second property of \Cref{lem:gprime} implies that $\phi' = \phi(H), A' = A(H)$, and that $\Delta'(u, v) \le \Delta(u, v \mid H)$ (since every shortest path in $G'$ is also a hop-minimal shortest path in $G$).
So we may continue
\begin{equation*}
\exopt(n, m+|H|, h') \ge \frac{\phi(H) - h' \cdot \left( \max \limits_{(u, v)} \Delta(u, v \mid H) \right)}{|A(H)|}.
\end{equation*}
Rearranging this inequality, we get
\begin{equation*}
\max \limits_{(u, v)} \Delta(u, v \mid H)  \ge \frac{\phi(H) - |A(H)|\cdot \exopt(n, m+|H|, h')}{h'}.
\end{equation*}
Note that the left-hand side of this inequality is the maximum attainable potential reduction in $G$, so our goal is now to simplify the right-hand side.
We next plug in the choice $h' \coloneqq  cm / log n$, and also observe that within the first $m$ rounds of the algorithm, we have $|H| \le m$ and therefore $\exopt(n, m+|H|, h') \le \exopt(n, 2m, h')$.
So we have
\begin{equation*}
\max \limits_{(u, v)} \Delta(u, v \mid H)  \ge \frac{\phi(H) - |A(H)|\cdot \exopt(n, 2m, \frac{cm}{\log n})}{cm / \log n}.
\end{equation*}
Next, note that $\phi(H)$ is defined as the sum of $|A(H)|$ terms, each of which is at least the hopbound parameter $2 \cdot \exopt(n, 2m, cm / \log n)$.
Thus it is at least twice the term $|A(H)| \cdot \exopt(n, 2m, cm/\log n)$ subtracted in the numerator.
So we may continue
\begin{equation*}
\max \limits_{(u, v)} \Delta(u, v \mid H)  \ge \frac{\Omega\left( \phi(H) \right)}{cm / \log n}.
\end{equation*}
This completes the proof of the potential reduction lemma.

\section{Combining Greedy and $\ell$-Covers to Obtain Slightly Smaller $n^{1/3}$-Shortcut Sets}
\label{sec:chain-shortcut-greedy}

In this section we show that a modified version of \Cref{alg:greedyhop} that is aware of an $\ell$-chain cover can be used to obtain shortcuts sets which are slightly better than the current best known bound: we achieve $\BigO(n^{1/3})$ hopbound using a $\BigO(n\log^{*} n)$ shortcut edges, compared to $\BigO(n \log^2 n)$ shortcut edges of \cite{kogan_algorithmic_2024}.

\begin{restatable}{theorem}{thmChainGreedy}
\label{thm:chain-greedy}
    \Cref{alg:chain-greedy} yields $\BigO(n^{1/3})$-shortcut sets of size $\BigO(n \log^* n)$.
\end{restatable}

\subsection{Background on Chain Covers \label{sec:prems}}
For a directed graph $G$ write $G^R$ for the \emph{reversed graph}, that is the graph where all edges from $G$ point in the opposite direction.
The \emph{transitive closure} $\tc(G)$ of a directed graph $G$ is the graph containing edge $(u,v)$ if and only if there is a path $u \leadsto v$ in $G$.
A \emph{chain} is a path in the transitive closure.

\subsubsection{Super-Shortcutting Paths.}
The following lemma is a slightly tighter version of Lemma~1.10 from \cite{kogan_beating_2022}.
\begin{lemma}[Lemma 1.3 of \cite{Raskhodnikova10}]
    \label{lem:path-shortcut}
    For any path graph $P_n$, there is a shortcut set of size $\BigO(n \log^* n)$ with hopbound $\hopbound \le 4$, which can be computed in $\BigO(n \log^* n)$ time.
\end{lemma}

The idea of \emph{supershortcuts} provide a way to shortcut a path to constant diameter.
This is a common subroutine for more general shortcutting algorithms.

\begin{corollary}
    \label{lem:supershortcut}
    For any DAG $G=(V,E)$ and any set of vertex-disjoint chains $\chains$ in $G$, there is a set of edges $H$ s.t. $\setsize{H} = \BigO(n \log^* n)$ and for any $C \in \chains$ and $u,v \in C$, we have $\dist_{G \cup H}(u,v) \le 4$.
    Write $\supershortcut(G)$ for such a set $H$.
    Such a set can be computed in $\BigO(H)$ time.
\end{corollary}

\begin{proof}
    For any chain $C = v_1, \dots, v_k$, include the edges $(v_i, v_{i+1})$ for $i \in [k-1]$ in $H$. 
    Now apply \Cref{lem:path-shortcut} to every chain and add the union of the obtained shortcut sets to $H$.
    This $H$ has the claimed properties.
\end{proof}

\subsubsection{$\ell$-Cover and $\ell$-Chain Covers.}
Kogan and Parter define the notion of an $\ell$-cover.
This structure is crucial to their approach.

\begin{definition}[$\ell$-Cover; Definition 2.1 of \cite{kogan_beating_2022}]
    For a given $n$-vertex directed graph $G = (V, E)$ and an integer parameter $\ell \in [n]$, an \emph{$\ell$-cover} for $G$ is a multiset of $\ell$ paths $\mathcal P = {P_1, \dots , P_\ell}$ (possibly consisting of single vertices or empty), satisfying the following: (1) $\sum_{P \in \mathcal P} \setsize{P} \le \min(\ell \cdot n, \diam(G) \cdot n)$, and (2) for any path $P \subseteq G$, it holds that $\setsize{V(P) \cap (V\setminus V(\mathcal P))} \le 2 n/ \ell$.
\end{definition}

For algorithmic reasons, we need the following derived notion.
Intuitively, in an $\ell$-cover vertices that are on more than one path only provide benefit (i.e., by covering a part of the graph) once, therefore we're interested in a vertex-disjoint version of the structure.
This yields a set of chains (instead of a multiset of paths).
Note this is a formalization of the output of the \texttt{DisjointChains} algorithm from \cite{kogan_beating_2022}.

\begin{definition}[$\ell$-Chain Cover]
\label{def:chain-cover}
    Let an \emph{$\ell$-chain cover} be a set of at most $\ell$ vertex-disjoint chains $\chains$ in $G$ such that for any path $P \subseteq G$, it holds that $\setsize{V(P) \cap (V\setminus V(\mathcal C))} \le 2 n/ \ell$.
\end{definition}

Plugging together two recent results, we can compute this object in almost-linear time.

\begin{lemma}
    \label{lem:fast-ell-cover}
    We can compute an $\ell$-chain cover $\chains$ in $\BigO(m^{1+o(1)})$ time.
\end{lemma}

\begin{proof}
    Applying the algorithm of C\'aceres \cite{caceres_minimum_2023} to extract a vertex-disjoint chain decomposition from the flow (which can be calculated using the algorithm from \cite{fast-deterministic-flow}) in the algorithm \texttt{PathCover} from \cite{kogan_beating_2022}, we directly obtain the desired result.
\end{proof}

For a chain $C$ and a node $u$, we often consider the first node $v$ on $C$ that is reachable from $u$.
We write $e(u,C) = (u,v)$ for this pair.
Note that $e(u,C)$ is in the transitive closure of the graph.

\subsection{Main Proof}

We now show Theorem \ref{thm:chain-greedy}.
We will recap some relevant prior work first.

\subsubsection{Recap: Kogan-Parter's $n^{1/3}$ Hopbound \cite{kogan_beating_2022}.}
Kogan and Parter introduced the first algorithm
(\texttt{Faster\allowbreak{}Shortcut\allowbreak{}Small\allowbreak{}Diam} in \cite{KP22a,kogan_beating_2022})
for finding shortcut sets of size $\Oish(n)$ for target hopbound $D = \Theta(n^{1/3})$.
Without loss of generality, the graph can be assumed to be a DAG (otherwise we can add a star on each strongly connected component and contract them).
Their algorithm proceeds as follows:
\begin{enumerate}
    \item 
Compute an $\ell$-chain cover $\chains$ for $\ell = 2n/D$ (see \cref{def:chain-cover}),
\item Super-shortcut every chain in $\chains$ (see \Cref{lem:supershortcut}),
\item Sample $\BigO(n \log n/D^2)$ chains $\chains'$ from $\chains$,
\item Sample $\BigO(n \log n/D)$ nodes $V'$ from $V$, and
\item Add a shortcut edge from every node in $V'$ to the earliest reachable node on every chain from $\chains'$.
\end{enumerate}
The $\ell$-chain cover ensures that \emph{any} path will have at most $D$ nodes not covered by chains. After super-shortcutting the chains, a path's length is thus roughly determined by the number of distinct chains it traverses. The sampling steps ensure that each sufficiently long path will, with high probability, be hit by a node in the $\BigO(D)$-length prefix and a chain in the $\BigO(D)$-length suffix, thereby reducing the diameter to $\BigO(D)$.

In this section, we show that steps 3-5 can be replaced by running a modified greedy algorithm that reduces a potential based on node-chain distance pairs.

\subsubsection{Running Greedy on Chain Covers.}
Our modified algorithm is based on an $\ell$-chain cover (with $\ell = 2n^{2/3}$) that we super-shortcut, but instead of sampling some chains and nodes and connecting them maximally, we use a greedy approach to connect the chains among each other.

\begin{figure}
    \centering
    \includegraphics[page=1]{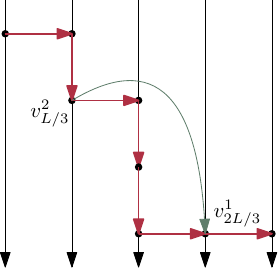}
    \caption{The (red) \emph{valid} path $P$ passes through $L$ chains in the $\ell$-chain cover. The shortcut edge between $v^{2}_{L/3}$ and $v^{1}_{2L/3}$ reduces the normalized distance between the first $L/3$ vertices on the path to the last $L/3$ chains on the path by $L/3$.}
    \label{fig:greedy-chains}
\end{figure}

Starting with an $\ell$-cover allows us to focus on only $n^{5/3}$ ``important'' pairs $(u,v)$ in the transitive closure, instead of all $n^{2}$ many. Let
\[S \coloneqq \{e(v,C) \mid v \in V, C \in \chains\}\]
be the set of important pairs between which we aim to shortcut. Recall that $e(v,C)$ is the edge $(v,w)$ in the transitive closure where $w$ is the earliest node on $C$ that $v$ can reach.
We call a $(u,v)$-path $P$ \emph{valid}
if for any chain $C\in \chains$ with $P\cap C\neq \emptyset$:
\begin{itemize}
\item $P\cap C$ forms a contiguous segment on the path $P$, and
\item if $w$ is the earliest node on $C$  that $u$ can reach in $G$,  the segment $P\cap C$ must start at $w$.
\end{itemize}
We consider the \emph{normalized distance} as the minimum number of chains that any \emph{valid} path between two points $(a,b) \in S$ must pass through (we denote this distance measure as $d'(a,b)$).
Recall that all chains are super-shortcut and thus any path contains at most $n^{1/3}$ nodes that are on no chain. Thus, if we reduce the \emph{normalized distance} between any pair in $S$ to $\BigO(n^{1/3})$, then the (usual) distance between any two nodes would also be $\BigO(n^{1/3})$.

We now define our modified potential $\phi'$ as
\[
    \phi' \coloneqq \sum_{(a,b) \in S} d'(a,b).
\]
Our modified greedy algorithm (see \Cref{alg:chain-greedy}) then chooses, as before, the edge (in the transitive closure) to add to the graph that reduces the potential $\phi'$ the most.
As in the warm-up for $\sqrt{n}$ hopbound, we will see that this reduces the potential by $\Omega(L^3)$ if $L$ is the length of the currently longest shortest path $P$, with respect to the \emph{normalized distance} $d'$. Indeed, say $P$ is such a path, and  assume it is passing though $L$ chains, where $v^1_i$ and $v^2_i$ are the first and last node on the $i$-th chain.
Then the edge $v^2_{L / 3} \to v^1_{2L/ 3}$ reduces the potential by $\Omega(L^3)$, see also \cref{fig:greedy-chains}.
See the appendix for a proof of the following lemma.

\begin{algorithm}[tbhp]
    \DontPrintSemicolon
    \SetKwFunction{FMain}{ChainCoverGreedy}
    \SetKwProg{Fn}{fun}{:}{}

    1. Compute an $\ell$-chain cover $\chains$ using \Cref{lem:fast-ell-cover}, with $\ell = 2n^{2/3}$.
    
    2. Super-shortcut $\chains$ using \Cref{lem:supershortcut}. 

    3. Repeatedly add edges to the graph that maximizes the potential reduction of $\phi'$, until the normalized distance of any pair in $S$ is at most $n^{1/3}$.

\caption{Compute a greedy shortcut set based on an $\ell$-chain cover.
    \label{alg:chain-greedy}}
\end{algorithm}

\begin{appendixlemmarep}
    \label{lem:chain-potential} Let $G$
    be a DAG and $\chains$ a $2n^{2/3}$-chain cover and suppose that any $d'$-shortest valid path for a pair in $S$ passes through at most $L$ chains. Then after adding $\BigO(n^{5/3}/L^2)$ greedy shortcut edges, any $d'$-shortest valid path for a pair in $S$ passes through at most $L / 2$ chains.
\end{appendixlemmarep}

\begin{appendixproof}
Since we have $n$ nodes, and $2n^{2/3}$ chains, we have that $|S|\le 2n^{5/3}$, and
    thus $\phi' \le 2n^{5/3} \cdot L$, where $L$ is the current longest shortest normalized distance of pairs in $S$.
    Observe that, by definition of \emph{valid} paths, any subpath of a $d'$-shortest valid path must also be a $d'$-shortest valid path.
    As in \cref{fig:greedy-chains}, the edge connecting $v^{2}_{L/3}$ to $v^{1}_{2L/3}$ would reduce the normalized distance of $\Omega(L^2)$ many pairs in $S$ (for each vertex on the prefix of the path, and each chain on the suffix of the path) by $\Omega(L)$. This relies on the fact that $v^{1}_{2L/3}$ is the earliest node on it's chain that is reachable from $v^{2}_{L/3}$, so shortcutting the valid path with this shortcut edge would also make a valid path.
    In particular, adding a greedy shortcut reduces the potential by at least $(L / (2 \cdot 3))^3 =  L^3 / 216$ as long as the longest shortest normalized distance is $\ge L/2$.
    Thus, the length of the longest $d'$-shortest path for a pair in $S$ must fall below $L / 2$ after adding at most $\frac{2n^{5/3}L}{L^3/256} = \BigO(n^{5/3}/L^2)$ edges.
\end{appendixproof}

This lemma in hand, the rest of the analysis is analogous to the warm-up and we get the following result.
\thmChainGreedy*


\begin{appendixproof}

First we argue that if the normalized distance of any pair in $S$ is at most $n^{1/3}$, then the (usual) distance of any two nodes $s$ and $t$ (where $s$ can reach $t$) is at most $\BigO(n^{1/3})$. Let $P$ be a shortest $(s,t)$-path. There are at most $n^{1/3}$ nodes on $P$ not on our chain cover. Thus, if the length is more than $n^{1/3}$, there must be some vertex $t'$ on our path that is also on a chain $C\in \chains$ such that $\dist(t', t)\le n^{1/3}$.
Since $e(s,C)$ is a pair in $S$, its normalized distance is at most $n^{1/3}$. This means that there is a (valid) path $P'$ between $s$ and $t$ that touches at most $n^{1/3}$ chains, and by the super-shortcutting we can assume this path uses at most $4$ edges inside each chain. In total, the $(s,t)$-distance is thus $\BigO(n^{1/3})$.

Now, we compute the number of shortcut edges used by \Cref{alg:chain-greedy}. The super-shortcutting step uses a total of $\BigO(n\log^* n)$ edges. The greedy potential reduction step uses, by \cref{lem:chain-potential}, $\BigO(n^{5/3}/L^{2})$ edges to half the maximum normalized distance from $L$ to $L/2$. We repeat until $L \le n^{1/3}$. By a geometric sum, this uses a total of 
\begin{equation*}
\sum_{L \in \{ n, \frac{n}{2}, \frac{n}{4}, \ldots, n^{1/3}\}}
\BigO(n^{5/3}/L^2) = 
\BigO(n^{5/3}/(n^{1/3})^2) = 
\BigO(n)
\end{equation*}
shortcut edges.
\end{appendixproof}


\bibliographystyle{plainurl}
\bibliography{papers.bib}

@article{kogan_beating_2022,
  title        = {Beating {Matrix} {Multiplication} for n{\textasciicircum}\{1/3\}-{Directed} {Shortcuts}},
  volume       = {229},
  copyright    = {Creative Commons Attribution 4.0 International license, info:eu-repo/semantics/openAccess},
  issn         = {1868-8969},
  url          = {https://drops.dagstuhl.de/entities/document/10.4230/LIPIcs.ICALP.2022.82},
  doi          = {10.4230/LIPICS.ICALP.2022.82},
  language     = {en},
  urldate      = {2024-11-04},
  journal      = {LIPIcs, Volume 229, ICALP 2022},
  author       = {Kogan, Shimon and Parter, Merav},
  collaborator = {Bojańczyk, Mikołaj and Merelli, Emanuela and Woodruff, David P.},
  year         = {2022},
  note         = {Artwork Size: 20 pages, 1031574 bytes
                  ISBN: 9783959772358
                  Medium: application/pdf
                  Publisher: Schloss Dagstuhl – Leibniz-Zentrum für Informatik},
  pages        = {82:1--82:20}
}

@misc{elkin_near-additive_2020,
  title     = {Near-{Additive} {Spanners} and {Near}-{Exact} {Hopsets}, {A} {Unified} {View}},
  copyright = {arXiv.org perpetual, non-exclusive license},
  url       = {https://arxiv.org/abs/2001.07477},
  doi       = {10.48550/ARXIV.2001.07477},
  urldate   = {2024-11-04},
  publisher = {arXiv},
  author    = {Elkin, Michael and Neiman, Ofer},
  year      = {2020},
  note      = {Version Number: 1}
}

@inproceedings{cao_efficient_2020,
  address   = {New York, NY, USA},
  series    = {{STOC} 2020},
  title     = {Efficient construction of directed hopsets and parallel approximate shortest paths},
  isbn      = {978-1-4503-6979-4},
  url       = {https://doi.org/10.1145/3357713.3384270},
  doi       = {10.1145/3357713.3384270},
  booktitle = {Proceedings of the 52nd {Annual} {ACM} {SIGACT} {Symposium} on {Theory} of {Computing}},
  publisher = {Association for Computing Machinery},
  author    = {Cao, Nairen and Fineman, Jeremy T. and Russell, Katina},
  year      = {2020},
  note      = {event-place: Chicago, IL, USA},
  pages     = {336--349}
}

@article{kogan_algorithmic_2024,
  title        = {The {Algorithmic} {Power} of the {Greene}-{Kleitman} {Theorem}},
  volume       = {308},
  copyright    = {Creative Commons Attribution 4.0 International license, info:eu-repo/semantics/openAccess},
  issn         = {1868-8969},
  url          = {https://drops.dagstuhl.de/entities/document/10.4230/LIPIcs.ESA.2024.80},
  doi          = {10.4230/LIPICS.ESA.2024.80},
  language     = {en},
  urldate      = {2024-11-21},
  journal      = {LIPIcs, Volume 308, ESA 2024},
  author       = {Kogan, Shimon and Parter, Merav},
  collaborator = {Chan, Timothy and Fischer, Johannes and Iacono, John and Herman, Grzegorz},
  year         = {2024},
  note         = {Artwork Size: 14 pages, 722270 bytes
                  ISBN: 9783959773386
                  Medium: application/pdf
                  Publisher: Schloss Dagstuhl – Leibniz-Zentrum für Informatik},
  pages        = {80:1--80:14}
}

@article{GHP20,
  title={Improved guarantees for vertex sparsification in planar graphs},
  author={Goranci, Gramoz and Henzinger, Monika and Peng, Pan},
  journal={SIAM Journal on Discrete Mathematics},
  volume={34},
  number={1},
  pages={130--162},
  year={2020},
  publisher={SIAM}
}

@misc{dinitz_approximation_2025,
 title={Approximation Algorithms for Optimal Hopsets},
  author={Dinitz, Michael and Koranteng, Ama and Nazari, Yasamin},
  booktitle={52nd International Colloquium on Automata, Languages, and Programming (ICALP 2025)},
  pages={69--1},
  year={2025},
  organization={Schloss Dagstuhl--Leibniz-Zentrum f{\"u}r Informatik}
}

@misc{caceres_maximum_2025,
  title      = {Maximum {Coverage} \$k\$-{Antichains} and {Chains}: {A} {Greedy} {Approach}},
  shorttitle = {Maximum {Coverage} \$k\$-{Antichains} and {Chains}},
  url        = {http://arxiv.org/abs/2502.06459},
  doi        = {10.48550/arXiv.2502.06459},
  urldate    = {2025-02-11},
  publisher  = {arXiv},
  author     = {Cáceres, Manuel and Grigorjew, Andreas and Jiamjitrak, Wanchote Po and Tomescu, Alexandru I.},
  month      = feb,
  year       = {2025},
  note       = {arXiv:2502.06459 [cs]}
}

@InProceedings{caceres_minimum_2023,
  author =	{C{\'a}ceres, Manuel},
  title =	{{Minimum Chain Cover in Almost Linear Time}},
  booktitle =	{50th International Colloquium on Automata, Languages, and Programming (ICALP 2023)},
  pages =	{31:1--31:12},
  series =	{Leibniz International Proceedings in Informatics (LIPIcs)},
  year =	{2023},
  volume =	{261}
}

@inproceedings{BP19,
  title={A trivial yet optimal solution to vertex fault tolerant spanners},
  author={Bodwin, Greg and Patel, Shyamal},
  booktitle={Proceedings of the 2019 ACM Symposium on Principles of Distributed Computing},
  pages={541--543},
  year={2019}
}

@article{Farshi14,
  title={Greedy spanner algorithms in practice},
  author={Farshi, Mohammad and HekmatNasab, MohammadJavad},
  journal={Scientia Iranica},
  volume={21},
  number={6},
  pages={2142--2152},
  year={2014},
  publisher={Sharif University of Technology}
}

@inproceedings{ABDKS20,
  title={Weighted additive spanners},
  author={Ahmed, Reyan and Bodwin, Greg and Sahneh, Faryad Darabi and Kobourov, Stephen and Spence, Richard},
  booktitle={International Workshop on Graph-Theoretic Concepts in Computer Science},
  pages={401--413},
  year={2020},
  organization={Springer}
}

@InProceedings{EGN21,
  author =	{Elkin, Michael and Gitlitz, Yuval and Neiman, Ofer},
  title =	{{Improved Weighted Additive Spanners}},
  booktitle =	{35th International Symposium on Distributed Computing (DISC 2021)},
  pages =	{21:1--21:15},
  series =	{Leibniz International Proceedings in Informatics (LIPIcs)},
  year =	{2021}
}

@inproceedings{FS16,
  title={The greedy spanner is existentially optimal},
  author={Filtser, Arnold and Solomon, Shay},
  booktitle={Proceedings of the 2016 ACM Symposium on Principles of Distributed Computing},
  pages={9--17},
  year={2016}
}

@article{AGU72,
  title={The transitive reduction of a directed graph},
  author={Aho, Alfred V. and Garey, Michael R and Ullman, Jeffrey D.},
  journal={SIAM Journal on Computing},
  volume={1},
  number={2},
  pages={131--137},
  year={1972},
  publisher={SIAM}
}

@InProceedings{ABP17,
  author        = {Amir Abboud and Greg Bodwin and Seth Pettie},
  title         = {A hierarchy of lower bounds for sublinear additive spanners},
  booktitle     = {Proceedings of the 28th Annual ACM-SIAM Symposium on Discrete Algorithms (SODA)},
  year          = {2017},
  organization  = {Society for Industrial and Applied Mathematics},
  pages         = {568--576},
}

@InProceedings{AB18,
  author        = {Amir Abboud and Greg Bodwin},
  title         = {Reachability Preservers: New Extremal Bounds and Approximation Algorithms},
  booktitle     = {Proceedings of the 29th Annual ACM-SIAM Symposium on Discrete Algorithms (SODA)},
  year          = {2018},
  pages			= {1865--1883},
  organization  = {Society for Industrial and Applied Mathematics},
}

@InProceedings{BDPV18,
  author        = {Greg Bodwin and Michael Dinitz and Merav Parter and Virginia Vassilevska Williams},
  title         = {Optimal Vertex Fault Tolerant Spanners (for fixed stretch)},
  booktitle     = {Proceedings of the 29th Annual ACM-SIAM Symposium on Discrete Algorithms (SODA)},
  pages			= {1884--1900},
  year          = {2018},
  organization  = {Society for Industrial and Applied Mathematics},
}

@inproceedings{HXX25,
  title={New Separations and Reductions for Directed Preservers and Hopsets},
  author={Hoppenworth, Gary and Xu, Yinzhan and Xu, Zixuan},
  booktitle={2025 ACM-SIAM Symposium on Discrete Algorithms (SODA)},
  year={2025},
  pages={4405--4443}
}

@inproceedings{BH23,
  title={Folklore sampling is optimal for exact hopsets: Confirming the $\sqrt n$ barrier},
  author={Bodwin, Greg and Hoppenworth, Gary},
  booktitle={2023 IEEE 64th Annual Symposium on Foundations of Computer Science (FOCS)},
  pages={701--720},
  year={2023},
  organization={IEEE}
}

@inproceedings{JLS19,
  title={Parallel reachability in almost linear work and square root depth},
  author={Jambulapati, Arun and Liu, Yang P and Sidford, Aaron},
  booktitle={2019 IEEE 60th Annual Symposium on Foundations of Computer Science (FOCS)},
  pages={1664--1686},
  year={2019},
  organization={IEEE}
}

@inproceedings{VXX24,
  title={Simpler and higher lower bounds for shortcut sets},
  author={Williams, Virginia Vassilevska and Xu, Yinzhan and Xu, Zixuan},
  booktitle={Proceedings of the 2024 Annual ACM-SIAM Symposium on Discrete Algorithms (SODA)},
  pages={2643--2656},
  year={2024},
  organization={SIAM}
}

@inproceedings{BW23,
  title={Closing the Gap Between Directed Hopsets and Shortcut Sets},
  author={Bernstein, Aaron and Wein, Nicole},
  booktitle={Proceedings of the 2023 Annual ACM-SIAM Symposium on Discrete Algorithms (SODA)},
  pages={163--182},
  year={2023},
  organization={SIAM}
}

@article{CE06,
  author        = {Coppersmith, Don and Elkin, Michael},
  title         = {Sparse sourcewise and pairwise distance preservers},
  journal       = {SIAM Journal on Discrete Mathematics},
  year          = {2006},
  volume        = {20},
  number        = {2},
  pages         = {463--501},
  publisher     = {SIAM},
}

@Article{PU89jacm,
  author		= {Peleg, David and Upfal, Eli},
  title			= {A trade-off between space and efficiency for routing tables},
  journal       = {Journal of the ACM (JACM)},
  year          = {1989},
  volume        = {36},
  number        = {3},
  pages         = {510--530},
  publisher     = {ACM},
}

@Article{DHZ00,
  author		= {Dorit Dor and Shay Halperin and Uri Zwick},
  title			= {All-Pairs Almost Shortest Paths},
  year			= {2000},
  journal		= {Siam Journal on Computing (SICOMP)},
  volume		= {29},
  number		= {5},
  pages			= {1740--1759}
}

@Article{BGJRW12,
  author    = {Bhattacharyya, Arnab and Grigorescu, Elena and Jung, Kyomin and Raskhodnikova, Sofya and Woodruff, David P},
  title     = {Transitive-closure spanners},
  journal   = {SIAM Journal on Computing},
  year      = {2012},
  volume    = {41},
  number    = {6},
  pages     = {1380--1425},
  publisher = {SIAM},
}

@Article{ADDJS93,
  author    = {Alth{\"o}fer, Ingo and Das, Gautam and Dobkin, David and Joseph, Deborah and Soares, Jos{\'e}},
  title     = {On sparse spanners of weighted graphs},
  journal   = {Discrete \& Computational Geometry},
  year      = {1993},
  volume    = {9},
  number    = {1},
  pages     = {81--100},
  publisher = {Springer},
}

@inproceedings{KP25,
  title={Having Hope in Missing Spanners: New Distance Preservers and Light Hopsets},
  author={Kogan, Shimon and Parter, Merav},
  booktitle={Proceedings of the 2025 Annual ACM-SIAM Symposium on Discrete Algorithms (SODA)},
  pages={4352--4374},
  year={2025},
  organization={SIAM}
}

@inproceedings{KP22,
  title={Having hope in hops: New spanners, preservers and lower bounds for hopsets},
  author={Kogan, Shimon and Parter, Merav},
  booktitle={2022 IEEE 63rd Annual Symposium on Foundations of Computer Science (FOCS)},
  pages={766--777},
  year={2022},
  organization={IEEE}
}

@InProceedings{ENS14,
  author       = {Elkin, Michael and Neiman, Ofer and Solomon, Shay},
  title        = {Light spanners},
  booktitle    = {International Colloquium on Automata, Languages, and Programming},
  year         = {2014},
  organization = {Springer},
  pages        = {442--452},
}

@InCollection{Raskhodnikova10,
  author    = {Raskhodnikova, Sofya},
  title     = {Transitive-closure spanners: A survey},
  booktitle = {Property testing},
  year      = {2010},
  publisher = {Springer},
  pages     = {167--196},
}

@article{BKMP10,
  author    = {Baswana, Surender and Kavitha, Telikepalli and Mehlhorn, Kurt and Pettie, Seth},
  title     = {Additive spanners and ($\alpha$, $\beta$)-spanners},
  journal   = {ACM Transactions on Algorithms (TALG)},
  year      = {2010},
  volume    = {7},
  number    = {1},
  pages     = {5},
  publisher = {ACM},
}

@InProceedings{Knudsen14,
  author       = {Knudsen, Mathias B{\ae}k Tejs},
  title        = {Additive spanners: A simple construction},
  booktitle    = {Scandinavian Workshop on Algorithm Theory},
  year         = {2014},
  organization = {Springer},
  pages        = {277--281},
}

@InProceedings{HP18,
  author =	{Shang-En Huang and Seth Pettie},
  title =	{{Lower Bounds on Sparse Spanners, Emulators, and Diameter-reducing shortcuts}},
  booktitle =	{16th Scandinavian Symposium and Workshops on Algorithm  Theory (SWAT 2018)},
  pages =	{26:1--26:12},
  series =	{Leibniz International Proceedings in Informatics (LIPIcs)},
  ISBN =	{978-3-95977-068-2},
  ISSN =	{1868-8969},
  year =	{2018},
  volume =	{101},
  editor =	{David Eppstein},
  publisher =	{Schloss Dagstuhl--Leibniz-Zentrum fuer Informatik},
  address =	{Dagstuhl, Germany},
  URL =		{http://drops.dagstuhl.de/opus/volltexte/2018/8852},
  URN =		{urn:nbn:de:0030-drops-88521},
  doi =		{10.4230/LIPIcs.SWAT.2018.26}
}

@article{UY91,
  title={High-probability parallel transitive-closure algorithms},
  author={Ullman, Jeffrey D and Yannakakis, Mihalis},
  journal={SIAM Journal on Computing},
  volume={20},
  number={1},
  pages={100--125},
  year={1991},
  publisher={SIAM}
}

@InProceedings{ABSHKS21,
  author =	{Ahmed, Reyan and Bodwin, Greg and Sahneh, Faryad Darabi and Hamm, Keaton and Kobourov, Stephen and Spence, Richard},
  title =	{{Multi-Level Weighted Additive Spanners}},
  booktitle =	{19th International Symposium on Experimental Algorithms (SEA 2021)},
  pages =	{16:1--16:23},
  year =	{2021}
}

@inproceedings{KP22a,
  title={New Diameter-Reducing Shortcuts and Directed Hopsets: Breaking the Barrier},
  author={Kogan, Shimon and Parter, Merav},
  booktitle={Proceedings of the 2022 Annual ACM-SIAM Symposium on Discrete Algorithms (SODA)},
  pages={1326--1341},
  year={2022},
  organization={SIAM}
}

@InProceedings{PST24,
  author =	{Petruschka, Asaf and Sapir, Shay and Tzalik, Elad},
  title =	{{Color Fault-Tolerant Spanners}},
  booktitle =	{15th Innovations in Theoretical Computer Science Conference (ITCS 2024)},
  pages =	{88:1--88:17},
  year =	{2024}
}

@article{Fineman19,
  title={Nearly work-efficient parallel algorithm for digraph reachability},
  author={Fineman, Jeremy T},
  journal={SIAM Journal on Computing},
  volume={49},
  number={5},
  pages={STOC18--500},
  year={2019},
  publisher={SIAM}
}

@inproceedings{BHT23,
  title={Bridge girth: A unifying notion in network design},
  author={Bodwin, Greg and Hoppenworth, Gary and Trabelsi, Ohad},
  booktitle={2023 IEEE 64th Annual Symposium on Foundations of Computer Science (FOCS)},
  pages={600--648},
  year={2023},
  organization={IEEE}
}

@InProceedings{Hesse03,
  author       = {Hesse, William},
  title        = {Directed graphs requiring large numbers of shortcuts},
  booktitle    = {Proceedings of the fourteenth annual ACM-SIAM symposium on Discrete algorithms},
  year         = {2003},
  organization = {Society for Industrial and Applied Mathematics},
  pages        = {665--669},
}

@Article{Thorup95,
  author    = {Thorup, Mikkel},
  title     = {Shortcutting planar digraphs},
  journal   = {Combinatorics, Probability and Computing},
  year      = {1995},
  volume    = {4},
  number    = {3},
  pages     = {287--315},
  publisher = {Cambridge University Press},
}

@Article{Kavitha17,
  author    = {Kavitha, Telikepalli},
  title     = {New pairwise spanners},
  journal   = {Theory of Computing Systems},
  year      = {2017},
  volume    = {61},
  number    = {4},
  pages     = {1011--1036},
  publisher = {Springer},
}

@inproceedings{BRR10,
  title={Finding sparser directed spanners},
  author={Berman, Piotr and Raskhodnikova, Sofya and Ruan, Ge},
  booktitle={IARCS Annual Conference on Foundations of Software Technology and Theoretical Computer Science (FSTTCS 2010)},
  pages={424--435},
  year={2010},
  organization={Schloss Dagstuhl--Leibniz-Zentrum fuer Informatik}
}

@INPROCEEDINGS{fast-deterministic-flow,
  author={Brand, Jan Van Den and Chen, Li and Kyng, Rasmus and Liu, Yang P. and Peng, Richard and Gutenberg, Maximilian Probst and Sachdeva, Sushant and Sidford, Aaron},
  booktitle={2023 IEEE 64th Annual Symposium on Foundations of Computer Science (FOCS)}, 
  title={A Deterministic Almost-Linear Time Algorithm for Minimum-Cost Flow}, 
  year={2023},
  volume={},
  number={},
  pages={503-514},
  doi={10.1109/FOCS57990.2023.00037}}

@article{chalermsook2025shortcuts,
  title={Shortcuts and Transitive-Closure Spanners Approximation},
  author={Chalermsook, Parinya and Jiang, Yonggang and Mukhopadhyay, Sagnik and Nanongkai, Danupon},
  journal={arXiv preprint arXiv:2502.08032},
  year={2025}
}

@inproceedings{LN2022,
  title={Near-Optimal Decremental Hopsets with Applications},
  author={{\L}{\k{a}}cki, Jakub and Nazari, Yasamin},
  booktitle={49th International Colloquium on Automata, Languages, and Programming (ICALP 2022)},
  year={2022},
  organization={Schloss Dagstuhl-Leibniz-Zentrum f{\"u}r Informatik}
}

@inproceedings{bernstein2021deterministic,
  title={Deterministic Decremental {SSSP} and Approximate Min-Cost Flow in Almost-Linear Time},
  author={Bernstein, Aaron and Probst Gutenberg, Maximilian and Saranurak, Thatchaphol},
  booktitle={62 Annual IEEE Symposium on Foundatios of Computer Science (FOCS 2022)},
  year={2021}
}

@inproceedings{chen2022maximum,
  title={Maximum flow and minimum-cost flow in almost-linear time},
  author={Chen, Li and Kyng, Rasmus and Liu, Yang P and Peng, Richard and Gutenberg, Maximilian Probst and Sachdeva, Sushant},
  booktitle={2022 IEEE 63rd Annual Symposium on Foundations of Computer Science (FOCS)},
  pages={612--623},
  year={2022},
  organization={IEEE}
}

@inproceedings{cao2023exact,
  title={Parallel exact shortest paths in almost linear work and square root depth},
  author={Cao, Nairen and Fineman, Jeremy T},
  booktitle={Proceedings of the 2023 Annual ACM-SIAM Symposium on Discrete Algorithms (SODA)},
  pages={4354--4372},
  year={2023},
  organization={SIAM}
}

@inproceedings{cao2020improved,
  title={Improved Work Span Tradeoff for Single Source Reachability and Approximate Shortest Paths},
  author={Cao, Nairen and Fineman, Jeremy T and Russell, Katina},
  booktitle={ACM Symposium on Parallelism in Algorithms and Architectures},
  year={2020}
}

@inproceedings{fineman2018nearly,
  title={Nearly work-efficient parallel algorithm for digraph reachability},
  author={Fineman, Jeremy T},
  booktitle={Proceedings of the 50th Annual ACM SIGACT Symposium on Theory of Computing},
  pages={457--470},
  year={2018}
}

@inproceedings{elkin2019RNC,
  title={Linear-size hopsets with small hopbound, and constant-hopbound hopsets in {RNC}},
  author={Elkin, Michael and Neiman, Ofer},
  booktitle={The 31st ACM Symposium on Parallelism in Algorithms and Architectures},
  pages={333--341},
  year={2019}
}

@article{elkin2019journal,
  title={Hopsets with constant hopbound, and applications to approximate shortest paths},
  author={Elkin, Michael and Neiman, Ofer},
  journal={SIAM Journal on Computing},
  year={2019},
  publisher={SIAM}
}

@inproceedings{elkin2017,
  title={Near-Optimal Distributed Routing with Low Memory},
  author={Elkin, Michael and Neiman, Ofer},
  booktitle={Proceedings of the ACM Symposium on Principles of Distributed Computing},
  year={2018},
  organization={ACM}
}

@article{censor2021,
  title={Fast approximate shortest paths in the congested clique},
  author={Censor-Hillel, Keren and Dory, Michal and Korhonen, Janne H and Leitersdorf, Dean},
  journal={Distributed Computing},
  volume={34},
  pages={463--487},
  year={2021},
  publisher={Springer}
}

@article{cohen2000,
  title={Polylog-time and near-linear work approximation scheme for undirected shortest paths},
  author={Cohen, Edith},
  journal={Journal of the ACM (JACM)},
  year={2000},
  publisher={ACM}
}

@inproceedings{henzinger2014,
  title={Decremental single-source shortest paths on undirected graphs in near-linear total update time},
  author={Henzinger, Monika and Krinninger, Sebastian and Nanongkai, Danupon},
  booktitle={2014 IEEE 55th Annual Symposium on Foundations of Computer Science},
  pages={146--155},
  year={2014},
  organization={IEEE}
}

@inproceedings{gutenberg2020,
  title={Deterministic algorithms for decremental approximate shortest paths: Faster and simpler},
  author={Gutenberg, Maximilian Probst and Wulff-Nilsen, Christian},
  booktitle={Proceedings of the Fourteenth Annual ACM-SIAM Symposium on Discrete Algorithms},
  pages={2522--2541},
  year={2020},
  organization={SIAM}
}

\appendix


\end{document}